\documentclass[12pt,english]{article}
\usepackage[T1]{fontenc}
\usepackage[latin9]{inputenc}
\usepackage{geometry}
\usepackage{color}
\usepackage{mathtools}
\usepackage{amsmath}
\usepackage{amsthm}
\usepackage{amssymb}
\usepackage{esint}

\usepackage{easyReview}				

\makeatletter
\theoremstyle{plain}
\newtheorem{thm}{\protect\theoremname}
\theoremstyle{remark}
\newtheorem{rem}[thm]{\protect\remarkname}

\theoremstyle{definition}
\newtheorem{example}[thm]{\protect\examplename}

\theoremstyle{remark}
\newtheorem{cor}[thm]{\protect\corollaryname}

\AtBeginDocument{

}\usepackage{babel}

\allowdisplaybreaks

\usepackage{natbib}         

\usepackage{hyperref}       

\usepackage{graphicx}				

\providecommand{\theoremname}{Theorem}

\providecommand{\remarkname}{Remark}

\makeatother

\usepackage{babel}
\providecommand{\remarkname}{Remark}
\providecommand{\theoremname}{Theorem}

\providecommand{\corollaryname}{Corollary}
\providecommand{\examplename}{Example}

\begin{document}
\title{Computing the Probability of a Financial Market Failure: A New Measure
of Systemic Risk}
\author{
    Robert Jarrow\\
    \small{S.C. Johnson Graduate School of Management}\\
    \small{Cornell University}\\
    \small{Ithaca NY, 14853}\\
    \small{ORCID 0000-0001-9893-9611}\\
    \small{raj15@cornell.edu}
    \and
    Philip Protter\thanks{Supported in part by NSF grant DMS-2106433} \\
    \small{Department of Statistics}\\
    \small{Columbia University}\\
    \small{New York, NY, 10027}\\
    \small{ORCID 0000-0003-1344-0403}\\
    \small{pep2117@columbia.edu}
    \and
    Alejandra Quintos\thanks{Supported in part by the Office of the Vice Chancellor for Research and Graduate Education at the University of Wisconsin-Madison with funding from the Wisconsin Alumni Research Foundation and by the Fulbright-Garc\'{i}a Robles Program} \thanks{Corresponding author}\\
    \small{Department of Statistics}\\
    \small{University of Wisconsin-Madison}\\
    \small{Madison, WI, 53706}\\
    \small{ORCID 0000-0003-3447-3255}\\
    \small{alejandra.quintos@wisc.edu}
}
\maketitle
\pagebreak

\begin{abstract}
This paper characterizes the probability of a market failure defined as the default of two or more globally systemically important banks (G-SIBs) in a small interval of time. The default probabilities of the G-SIBs are correlated through the possible existence of a market-wide stress event. The characterization employs a multivariate Cox process across the G-SIBs, which allows us to relate our work to the existing literature on intensity-based models. Various theorems related to market failure probabilities are derived, including the probability of a market failure due to two banks defaulting over the next infinitesimal interval, the probability of a catastrophic market failure, the impact of increasing the number of G-SIBs in an economy, and the impact of changing the initial conditions of the economy's state variables. We also show that if there are too many G-SIBs, a market failure is inevitable, i.e., the probability of a market failure tends to 1.
\end{abstract}
\section{Introduction and Summary}

For regulators, characterizing the probability of a financial market
failure, or systemic risk, is important because such a characterization
enables them to understand how their regulatory actions affect its
magnitude. In this regard, numerous systemic risk measures have been
proposed in the literature, each with associated benefits and limitations.
For literature reviews of the existing collection of systemic risk
measures, see \cite{Bisias.Flood.Lo.2012}
and \cite{Engle.2018}. This paper provides another measure
of systemic risk, different from the existing set. According to the
systemic risk measure taxonomies in \cite{Bisias.Flood.Lo.2012},
ours is a macroeconomic or macroprudential measure, which is a based
on a default intensity model. As such, it is a forward-looking measure,
which satisfies the following characteristics: 
\begin{enumerate}
\item it is consistent with the economic theories relating to the causes
of financial market failures (macroeconomic), 
\item it uses the existing regulatory designations of globally systemically
important banks (G-SIBs), financial institutions that are ``too big
to fail'' (macroeconomic), 
\item it can be estimated using existing hazard rate methodologies (default
intensity), and 
\item it facilitates quantifying the impact of regulatory policy changes
on systemic risk (macroprudential). 
\end{enumerate}
Our measure of systemic risk is the\emph{ probability that any two
G-SIBs default at the ``same time.'' }A G-SIB is any financial institution
that has been designated by the Financial Stability Board (FSB) as
large enough such that if it fails, its failure affects the health
of the financial system. Operationally, a bank is designated as a
G-SIB if various indicators of its financial health, in aggregate,
exceed some threshold (see FSB \citeyear{FSB.2020} and BIS \citeyear{BIS.2014}).
There were 30 such G-SIBs designated by the FSB in 2020. By the ``same
time'' we mean within a short time period of each other, say 1 week.

The idea underlying our measure is that if one G-SIBs fails, regulators
can manage the resulting crisis to ensure that a market wide failure
does not occur. Examples of such past episodes include the failure
of Long Term Capital Management in 1998 and Lehman Brothers, together with Bear Sterns, in 2008.
For both of these episodes, regulators were able to manage the crisis
and prevent a market-wide failure. However, if two (or more) G-SIBs
fail within a short time period of each other, then our measure asserts
that the crisis is uncontrollable by regulators and the market fails.

Our measure is consistent with economic theories of market failures
because, in a reduced form fashion, it implicitly includes the causes
for the failure, e.g. the ``drying-up'' of short-term funding, the
bursting of an asset price bubble, or the propagation of defaults
in a network of banks due to inter-linked funding (see \cite{Allen.Carletti.2013}, 
\cite{Acemoglu.2015}, \cite{Jarrow.Lam.2019}). And, it also explicitly
incorporates the marginal impact of a G-SIBs's default on the probability
of a financial market failure (see \cite{Acharya.Ped.2009} for related discussion).

By its definition, our measure builds upon the fundamental analysis
already done by regulators in identifying financial institutions that
are G-SIBs. In the identification of G-SIBs, regulators
include public information (market prices, macroeconomic statistics,
a financial institution's annual reports), non-public information
available via the regulatory channel, and expert judgement (see BIS
\citeyear{BIS.2014}). The use of the G-SIB designations as a basis
for our systemic risk measure, which include this non-public information
and expert judgement, yields an additional benefit not available with
the use of publicly available information alone.

Our measure can also be estimated due to its construction, because
the probability of a market failure incorporates the existing marginal
probabilities of a G-SIB defaulting. These probabilities can be obtained
as in the existing hazard rate estimation literature, see \cite{Chava.Jarrow.2004}, 
\cite{Campbell.2008}, \cite{Shumway.2001}. The
data necessary to compute these marginal default probabilities is
publicly available, and they consist of historical data on financial
institutions defaults, annual reports (balance sheet data), and market
variables (prices and macroeconomic statistics). Finally, given the analytic representation
of our systemic risk measure, it is easy to compute the impact of
a regulatory policy change on the probability of a market failure,
e.g., such regulatory actions might be the breaking-up of a G-SIB
or the increase in a G-SIB's capital. These regulatory changes correspond
to modifying various input variables underlying the market failure
probabilities and determining their impact on the resulting value.

Our paper falls within the credit risk literature studying correlated
firm defaults. There are two approaches in this literature. One approach
is to have the firms' default intensities dependent on each other,
either through information events or counterparty risk. This method
typically uses Cox processes with independent default indicators implying
simultaneously occurring default times happen with zero probability.
The second approach extends these models to allow for a strictly positive
probability of simultaneous defaults. Our paper lies within the second
class. Papers in both of these genres include \cite{Giesecke.2003}, 
\cite{Lindskog.McNeil}, \cite{Bielecki.Cousin.Crepey.Herbertsson}, 
\cite{Bielecki.Rutkowski}, \cite{Brigo.Pallavicini.Torresetti.2}, 
(\citeyear{Brigo.Pallavicini.Torresetti.1}),
\cite{Coculescu}, \cite{ElKaroui.Jeanblanc.Jiao}, and \cite{Liang.Wang}.

The mathematical model in our paper can be viewed as a special case
of that presented in \cite{Bielecki.Cousin.Crepey.Herbertsson}.
A key difference is that we apply our model to characterize systemic
risk, while \cite{Bielecki.Cousin.Crepey.Herbertsson} study the
pricing of basket credit derivatives, in particular credit default
obligations (CDOs). In addition, as previously noted, we derive a
collection of results related to the probability of a systemic risk
event occurrence that are new to this literature.

An outline for the paper is as follows. Section 2 presents the model,
while section 3 contains the key theorems. Section 4 provides comparative
statics and concludes the paper.

\section{The Model}

The following model is based on \cite{Protter.Quintos.2021}.
Fix a filtered probability space $\left(\Omega,\mathcal{F},\mathbb{P},\mathbb{F}\right)$
satisfying the usual conditions and large enough to support an $\mathbb{R}^{d}$
- valued right continuous with left limits existing stochastic process
$X=\{X_{t},t\geq0\}$ and $K+1$ independent exponential random variables
each with parameter $1$, i.e. $(Z_{i},i=0,...,K)$. 
It must be noted that the $K+1$ exponential random variables are independent of each other and
of the stochastic process $X$.

Consider a financial market that contains $i=1,...,K$ financial institutions
that are classified as G-SIBs, i.e. too big to fail. There can be
numerous other financial institutions in the market, but their existence
will not be explicitly included in our systemic risk measure. However,
these non-G-SIBs are implicitly included as will be subsequently noted.

The stochastic process $X$ represents a vector of state variables
characterizing the health of the economy and the $K$ G-SIBs. It includes
macro variables such as the inflation rate, the unemployment rate,
the level of interest rates, and G-SIB specific balance sheet quantities
such as their capital ratios.

\subsection{The G-SIBs' Default Times due to Idiosyncratic Events}

Define the default time for the $i^{th}$ G-SIB due to i\emph{diosyncratic
events} as 
\[
\eta_{i}\coloneqq\inf\{s:A_{i}(s)\geq Z_{i}\}
\]
where $A_{i}(s)=\intop_{0}^{s}\alpha_{i}(X_{r})dr$ and $Z_{i}\sim\mathrm{Exp}(1)$.
Assume that $\lim_{s\rightarrow \infty} A_i(s) = \infty$ and, to fix notation, let 
$f_i(x):=  \alpha_i(X_{x})e^{-A_i(x)}$. That is, $f_i(x)$ is the density of $\eta_{i}$ given
the sigma-field generated by $X$ over $[0,\infty)$ denoted as 
$\mathcal{F}^{X}_\infty:=\sigma\left(X_{t}:t\in[0,\infty)\right)$. The condition 
$\lim_{s\rightarrow \infty} A_i(s) = \infty$ ensures that $f_i$ is a proper density, i.e., 
it integrates to 1. 

The process $\alpha_{i}(\cdot):\mathbb{R}^{d}\rightarrow[0,\infty)$
is the default intensity of the $i^{th}$ G-SIB dependent upon the
state variable process $X$. The default intensity is assumed to be
a non-random, positive, continuous function. This implies that $A_{i}(s)$
are continuous and strictly increasing for any $s\geq0$. We note
that the default intensity of the $i^{th}$ G-SIB can be estimated
using standard hazard rate estimation techniques as \cite{Chava.Jarrow.2004}, 
\cite{Campbell.2008}, \cite{Shumway.2001}.

An idiosyncratic event causing default for a G-SIB is one that is
unique to the bank, after conditioning on the state variable process
$X$. For example, it could be due to fraudulent trades by a rogue
trader or incompetent management. As defined, by construction, the
idiosyncratic event default times of the G-SIBs are Cox processes,
conditionally independent across G-SIBs given $\mathcal{F}^{X}_\infty.$
This implies that $\mathbb{P}\left(\eta_{i}=\eta_{j}\right)=0$ for
$i\neq j$. More explicitly, by taking an expectation on the following
expression, 
\begin{align*}
\mathbb{P}\left(\eta_{i}=\eta_{j}|\mathcal{F}^{X}_\infty\right) & =\int_{0}^{\infty}\int_{0}^{\infty}1_{\{x=y\}}f_{{i}}(x)f_{{j}}(y)dydx\\
 & =\int_{0}^{\infty}\mathbb{P}\left(\eta_{j}=x|\mathcal{F}^{X}_\infty\right)f_{{i}}(x)dx=0
\end{align*}
where $f_{{i}}(x)$ and $f_{{j}}(y)$ are the continuous densities 
of $\eta_{i}$ and $\eta_{j}$ given $\mathcal{F}^{X}_\infty$.
\begin{example}
(\emph{Destructive Competition})\label{exa:(Destructive-Competition)}

A useful example of an idiosyncratic default intensity is one that
depends on the number of G-SIBs, i.e. $A_{i}(t,K):=\int_{0}^{t}\alpha_{i}(X_{u},K)du$,
where the marginal probability of a default increases with $K$ for
each $i=1,2,\dots,K$.

The interpretation is that as the number of G-SIBs increase, the banks
compete more aggressively with each other to maintain market share
and profitability. In doing so, they take on riskier investments to
increase expected returns, which in turn, increases idiosyncratic
default risk. This is in fact what occurred prior to the credit crisis
of 2007 when financial institutions invested in riskier AAA rated
collateralized debt obligations (CDOs) instead of the riskless AAA
rated U.S. Treasuries to obtain increased yields
(see \cite{Crouhy.2009} or \cite{Protter.2009} for a more detailed
explanation).

A special case of this intensity is when $\alpha_{i}(X_{t},K)=\ln(K)+b(X_{t})$
for an appropriate measurable and integrable $b(\cdot):\mathbb{R}^{d}\rightarrow[0,\infty)$.
Then, 
\[
A_{i}(t,K):=\int_{0}^{t}\alpha_{i}(X_{u},K)du=t\ln(K)+\int_{0}^{t}b(X_{u})du=t\ln(K)+B(t).
\]
\end{example}

\subsection{Market-Wide Stress Events}

Next, define $\eta_{0}$ to be the occurrence of a market-wide stress
event, as distinct from an idiosyncratic default event specific to
a single G-SIB. For example, it could be the drying up of short term
funding markets, the bursting of an asset price bubble (as in the
housing market prior to the 2007 credit crisis), or a large collection
of non G-SIBs defaulting in a short period of time. This market-wide
stress event implicitly includes the influence of the remaining non-G-SIBs
in the market, and their inter-relationships among themselves and
the G-SIBs.

Define the first time that a market-wide stress event occurs as 
\[
\eta_{0}\coloneqq\inf\{s:A_{0}(s)\geq Z_{0}\}
\]
where $A_{0}(s)=\intop_{0}^{s}\alpha_{0}(X_{r})dr$, $Z_{0}\sim\mathrm{Exp}(1)$,
and $\alpha_{0}(\cdot):\mathbb{R}^{d}\rightarrow[0,\infty)$ is a
non-random, positive, continuous function. 
Assume that $\lim_{s\rightarrow \infty} A_0(s) = \infty$ and let 
$f_0(x):=  \alpha_0(X_{x})e^{-A_0(x)}$. That is, $f_0(x)$ is the density of $\eta_{0}$ given
$\mathcal{F}^{X}_\infty$.

Note that, because of the
conditional independence assumption given $\mathcal{F}^{X}_\infty$, we have
$\mathbb{P}\left(\eta_{i}=\eta_{0}\right)=0$ for all $i$. Here the
intensity process of a market-wide stress event, $\alpha_{0}(X)$,
probably cannot be estimated using historical time series data given
the infrequency of their occurrence. However, a financial institution
or regulator can use expert judgement to facilitate the practical
computation of this quantity.

\subsection{The G-SIBs' Default Times}

Finally, we define the default time of the $i^{th}$ G-SIB as 
\[
\tau_{i}=\min(\eta_{0},\eta_{i})
\]
for $i=1,...,K$. This is the first time that either an $i^{th}$ G-SIB
defaults due to an idiosyncratic event or a market-wide stress event
occurs. Note that this definition implicitly characterizes the market-wide
stress event as one which is catastrophic enough to cause G-SIBs to
default on their obligations.

Given the above structure, we have that the survival probability of
the $i^{th}$ G-SIB is
\begin{equation}
\mathbb{P}\left(\tau_{i}>t\right)=\mathbb{E}\left[\exp\left(-A_{0}(t)-A_{i}(t)\right) \right].
\end{equation}

\begin{rem}
(\emph{Destructive Competition})

For the special case of destructive competition (see Example \ref{exa:(Destructive-Competition)}),
we have that the $i^{th}$ G-SIB's survival probability $\mathbb{P}\left(\tau_{i}>t\right)=\mathbb{E}\left[e^{-A_{0}(t)-A_{i}(t,K)} \right] $
is decreasing as $K$ increases. This implies, of course, that as the number
of G-SIB's increases, the probability of any single G-SIB defaulting
increases.
\end{rem}

\subsection{The Market Failure Time}

We now can define a financial market failure. To fix the intuition,
as an initial attempt, we first define a financial market failure
as the event 
\[
\left\{ \omega\in\Omega:\tau_{i}=\tau_{j}\;\mathrm{for\:some}\:(i,j)\in\left(1,...,K\right)\times\left(1,...,K\right),i\neq j\right\} ,
\]
and the probability of a financial market failure, our systemic risk
measure, as 
\[
\mathbb{P}\left(\tau_{i}=\tau_{j}\;\mathrm{for\:some}\:(i,j)\in\left(1,...,K\right)\times\left(1,...,K\right),i\neq j\right).
\]
This is the probability that two G-SIBs default at the exactly the
same time. The idea underlying this market-failure probability is
that if one G-SIB fails, regulators can manage the resulting failure
to ensure that a market-wide failure does not occur. However, if two
(or more) G-SIBs fail at the same time, then such an event is uncontrollable
by the regulators, resulting in a market-wide failure.

Unfortunately, there is a problem with this initial systemic risk
measure. Given the definition of the $i^{th}$ G-SIB's default time
$\tau_{i}$ and the conditional independence assumption given $\mathcal{F}^{X}_\infty$
across $\eta_{i}$ for $i=0,1,...,K$, a market failure occurs under
this definition with probability one if and only if $\eta_{0}\leq\eta_{i}$
and $\eta_{0}\leq\eta_{j}$ for some pair $i\neq j$. This is because
$\mathbb{P}\left(\eta_{i}=\eta_{j}\right)=0$ for $i\neq j$, $(i,j)\in\left(1,...,K\right)\times\left(1,...,K\right)$.
In essence, a market failure only occurs under this definition, in
probability, when a market-wide stress event occurs. In probability,
the existence of G-SIBs is irrelevant to this initial systemic risk
measure. To remove this problem, we generalize the definition of a
financial market failure event to be 
\[
\left\{ \omega\in\Omega:\left|\tau_{i}-\tau_{j}\right|<\varepsilon\;\mathrm{for\:some}\:(i,j)\in\left(1,...,K\right)\times\left(1,...,K\right),i\neq j\right\} 
\]
for a given $\varepsilon>0$, and our (final) systemic risk measure
as 
\[
\mathbb{P}\left(\left|\tau_{i}-\tau_{j}\right|<\varepsilon\;\mathrm{for\:some}\:(i,j)\in\left(1,...,K\right)\times\left(1,...,K\right),i\neq j\right).
\]
Under this systemic risk measure, a market failure can occur for two
reasons: a market-wide stress event occurs, or two G-SIBs experience
idiosyncratic default events within an $\varepsilon$ time period
of each other. The next section characterizes this market-wide default
probability.

\section{Theorems}

This section provides the key theorems characterizing the probability
distribution of defaults times for the various G-SIBs (``banks'')
and the probability of a market failure.

For easiness of notation, through the rest of this paper, let:
\begin{equation}
	\tilde{\mathbb{P}}(\cdot) := \mathbb{P}\left(\cdot | \mathcal{F}^{X}_\infty \right).
\end{equation}
This implies, for example, that $\tilde{\mathbb{P}}(\eta_j > t) = \exp \left(-A_j(t)\right)$ as we now show by recalling
the definition of $\eta_j$:
\begin{equation*}
	\tilde{\mathbb{P}}(\eta_j > t) =\tilde{\mathbb{P}} \left(Z_j > A_j(t)\right) = \exp \left(-A_j(t)\right).
\end{equation*}
Also, let $P_K$ stand for the set of all possible permutations of $(1,2,\dots, K)$. Hence, if $j\in P_K$, then
	$j=\left(j_{1}=\sigma(1),j_{2}=\sigma(2),\dots,j_{K}=\sigma(K)\right)$.

\subsection{The Joint Distribution of Banks' Default Times}

Our first theorem characterizes the joint probability distribution
of the banks' default times $\left(\tau_{1},\tau_{2},\dots,\tau_{K}\right)$. 
\begin{thm}
    \label{Theorem.Joint.Distribution} (Joint Distribution $\left(\tau_{1},\tau_{2},\dots,\tau_{K}\right)$).
    \begin{multline}
        \mathbb{P}\left(\tau_{1}>t_{1},\tau_{2}>t_{2},\dots,\tau_{K}>t_{K}\right)= \\
        \mathbb{E}\left[\exp\left(-\sum\limits_{i=1}^{K}A_{i}(t_{i})-A_{0}\left(\max\left(t_{1},t_{2},\dots,t_{K}\right)\right)\right)\right]. \label{eq: joint distribution}   
    \end{multline}
\end{thm}

\begin{proof}
Let $M:=\max\left(t_{1},t_{2},\dots,t_{K}\right)$.
Then,
\begin{multline*}
    \mathbb{P}\left(\tau_{1}>t_{1},\tau_{2}>t_{2},\dots,\tau_{K}>t_{K}| \left(X_{u}\right)_{ u\leq M} \right) =\mathbb{P}\left(\bigcap_{i=1}^{K}\eta_{i}\wedge\eta_{0}>t_{i}| \left(X_{u}\right)_{ u\leq M}\right)\\
    =\mathbb{P}\left(\eta_{1}>t_{1},\dots,\eta_{K}>t_{K},\eta_{0}>M| \left(X_{u}\right)_{ u\leq M} \right)\\
    =\exp\left(-\sum\limits _{i=1}^{K}A_{i}(t_{i})-A_{0}\left(M\right)\right).
\end{multline*}
The last equality follows from the mutual conditional independence assumption
of $\eta_{1},\eta_{2},\dots,\eta_{K},\eta_{0}$.
We can conclude the theorem by taking an additional expectation. 
\end{proof}
This distribution is a multivariate version of the Cox process, that is,
marginally each bank's default time is a Cox process. However, the
default times across the banks are not independent. The difference
in the joint distribution, relative to a standard Cox process, is
due to the last term in the exponent, $A_{0}\left(\max\left(t_{1},t_{2},\dots,t_{K}\right)\right)$,
which depends on the distribution of the first $K$ default times
exceeding the given times $t_{1},t_{2},\dots,t_{K}$. The form of
this multivariate distribution is tractable, facilitating subsequent
computations.
\begin{rem}
(\emph{Destructive Competition})

When there is destructive competition (see Example \ref{exa:(Destructive-Competition)}),
the joint survival probability of the $K$ banks is decreasing in
the number of G-SIBs, i.e. 
\begin{multline*}
    \mathbb{P}\left(\tau_{1}>t_{1},\tau_{2}>t_{2},\dots,\tau_{K}>t_{K}\right)=\\
    \mathbb{E}\left[\exp\left(-\sum\limits_{i=1}^{K}A_{i}(t_{i},K)-A_{0}\left(\max\left(t_{1},t_{2},\dots,t_{K}\right)\right)\right)\right]
\end{multline*}
is decreasing as $K$ increases.
\end{rem}

\medskip{}

\begin{cor}
(\emph{Constant Default Intensities}) \par
If $\alpha_i(X_t) = \alpha_i$ for all $t\geq 0$, all $i=1, 2, \dots, K$, and where $\alpha_i\in \mathbb{R}^+$, then:
\begin{multline}
    \mathbb{P}\left(\tau_{1}>t_{1},\tau_{2}>t_{2},\dots,\tau_{K}>t_{K}\right) \\
    =\exp\left[-\sum\limits _{i=1}^{K}\alpha_{i}t_{i}-\alpha_{0}\max\left(t_{1},t_{2},\dots,t_{K}\right)\right].
\end{multline}
\end{cor}

Next, we can deduce the probability of a market failure. 

\subsection{The Market Failure Probability}

This is the key theorem in our paper. 
\begin{thm}
(Market Failure Probability) \label{Thm Market Failure Probability}\textcolor{black}{{}
\begin{multline}
    \mathbb{P}\left(\left|\tau_{i}-\tau_{j}\right|<\varepsilon\;for\:some\:(i,j)\in\left(1,...,K\right)\times\left(1,...,K\right),{\color{red}{\normalcolor i\neq j}}\right)=\\
    {\normalcolor 1-\mathbb{E}\Biggl[\sum_{j\in P_K}\int\limits _{0}^{\infty}f_{j_{1}}(x_{1})\int\limits _{x_{1}+\varepsilon}^{\infty}f_{j_{2}}(x_{2})\int\limits _{x_{2}+\varepsilon}^{\infty}f_{j_{3}}(x_{3})\dots\int\limits _{x_{K-3}+\varepsilon}^{\infty}f_{j_{K-2}}(x_{K-2})}\\
    \int\limits _{x_{K-2}+\varepsilon}^{\infty}f_{j_{K-1}}(x_{K-1})\exp\left[-A_{j_{K}}(x_{K-1}+\varepsilon)-A_{0}(x_{K-1}+\varepsilon)\right]\,dx_{K-1}\,dx_{K-2}\dots\\
    \dots\,dx_{3}\,dx_{2}\,dx_{1}\Biggr]. \label{eq: result}
\end{multline}
}
\end{thm}

\begin{proof}
For $k=1,\dots,K$, let $\tau_{(k)}$ be the $k^{th}$ order statistic
of $\left(\tau_{1},\tau_{2},\dots,\tau_{K}\right)$. For example,
$\tau_{(1)}=\min\left(\tau_{1},\tau_{2},\dots,\tau_{K}\right)$ and
$\tau_{(K)}=\max\left(\tau_{1},\tau_{2},\dots,\tau_{K}\right)$. Similarly,
for $k=1,2,\dots,K,K+1$, let $\eta_{(k)}$ be the $k^{th}$ order
statistic of $\left(\eta_{0},\eta_{1},\eta_{2},\dots,\eta_{K}\right)$.
For example, $\eta_{(1)}=\min\left(\eta_{0},\eta_{1},\eta_{2},\dots,\eta_{K}\right)$
and $\eta_{(K+1)}=\max\left(\eta_{0},\eta_{1},\eta_{2},\dots,\eta_{K}\right)$.

Note that by construction $\eta_{0},\eta_{1},\eta_{2},\dots,\eta_{K}$
are mutually independent.

The complement of the event $\left\{ \left|\tau_{i}-\tau_{j}\right|<\varepsilon\;\mathrm{for\:some}\:(i,j),\:i\neq j\right\} $ is equal to $\left\{ \left|\tau_{i}-\tau_{j}\right|\geq\varepsilon\;\mathrm{for\:all}\:(i,j),\:i\neq j\right\} $ and hence,
\begin{multline*}
    \tilde{\mathbb{P}} \left(\left|\tau_{i}-\tau_{j}\right|\geq\varepsilon\;\mathrm{for\:all}\:(i,j),\:i\neq j\right) \\
    = \tilde{\mathbb{P}}\left(\tau_{\left(K\right)}\geq\tau_{\left(K-1\right)}+\varepsilon,\tau_{\left(K-1\right)}\geq\tau_{\left(K-2\right)}+\varepsilon,\dots,\tau_{\left(2\right)}\geq\tau_{\left(1\right)}+\varepsilon\right)\\
    =\tilde{\mathbb{P}}\left(\eta_{\left(K+1\right)}\geq\eta_{\left(K\right)}\geq\eta_{\left(K-1\right)}+\varepsilon,\eta_{\left(K-1\right)}\geq\eta_{\left(K-2\right)}+\varepsilon,\dots \right.\\
    \left.\dots,\geq\eta_{\left(1\right)}+\varepsilon,\eta_{0}\geq\eta_{\left(K\right)}\right).
\end{multline*}
If $\eta_{0}$ was not greater or equal to $\eta_{(K)}$, then there
exists some $k<K$ such that $\eta_{(k)}=\eta_{0}$. This implies
that $\tau_{(k)}=\tau_{(k+1)}=\dots=\tau_{(K)}$, which is the situation
we wish to avoid.

Now, we can work with the $K+1$ independent random variables, i.e.,
$\left(\eta_{0},\eta_{1},\eta_{2},\dots,\eta_{K}\right)$. We just
need to compute the distribution of the order statistics, which follows
by the usual construction of taking the sum over all possible permutations.
Let us start by finding the number of possible permutations. Because
of the restriction $\eta_{0}\geq\eta_{\left(K\right)}$, there are
$2\left(K!\right)$ permutations and then, 
\begin{multline}
    \tilde{\mathbb{P}} \left(\eta_{\left(K+1\right)}\geq\eta_{\left(K\right)}\geq\eta_{\left(K-1\right)}+\varepsilon,\eta_{\left(K-1\right)}\geq\eta_{\left(K-2\right)}+\varepsilon,\dots\right.\\
    \left.\dots,\eta_{\left(2\right)}\geq\eta_{\left(1\right)}+\varepsilon,\eta_{0}\geq\eta_{\left(K\right)}\right)=\sum_{j\in\tilde{P}}\int\limits _{0}^{\infty}f_{j_{1}}(x_{1})\int\limits _{x_{1}+\varepsilon}^{\infty}f_{j_{2}}(x_{2})\dots \\
    \dots \int\limits _{x_{K-1}+\varepsilon}^{\infty}f_{j_{K}}(x_{K})\int\limits _{x_{K}}^{\infty}f_{j_{K+1}}(x_{K+1})\,dx_{K+1}\,dx_{K}\dots\,dx_{2}\,dx_{1}\label{eq: (1)}
\end{multline}
where the sum is taken over all $2(K!)$ possible permutations such
that 
\begin{align*}
    j=\left(j_{1}=\sigma(1),j_{2}=\sigma(2),\dots,j_{K}=\sigma(K),j_{K+1}=0\right)
    \intertext{or}
    j=\left(j_{1}=\sigma(1),j_{2}=\sigma(2),\dots j_{K}=0,j_{K+1}=\sigma(K)\right).
\end{align*}
We can simplify (\ref{eq: (1)}). Let us fix one of the permutations
$j=\left(j_{1},j_{2},\dots,j_{K+1}\right)$ of $K$ such that
$j_{1}=\sigma(1),j_{2}=\sigma(2),\dots,j_{K}=\sigma(K),j_{K+1}=0$
and let us focus on the 2 innermost integrals,

\begin{multline} \label{eq: (2)}
    \int\limits _{x_{K-1}+\varepsilon}^{\infty}f_{j_{K}}(x_{K})\int\limits _{x_{K}}^{\infty}f_{j_{K+1}}(x_{K+1})\,dx_{K+1}\,dx_{K}\\
    =\int\limits _{x_{K-1}+\varepsilon}^{\infty}f_{j_{K}}(x_{K})\int\limits _{x_{K}}^{\infty}\alpha_{j_{K+1}}\left(X_{x_{K+1}}\right)e^{-A_{j_{K+1}}\left(x_{K+1}\right)}\,dx_{K+1}\,dx_{K}\\
    =\int\limits _{x_{K-1}+\varepsilon}^{\infty}\alpha_{j_{K}}\left(X_{x_{K}}\right)e^{-A_{j_{K}}\left(x_{K}\right)-A_{j_{K+1}}\left(x_{K}\right)}\,dx_{K}\\
    =\int\limits _{x_{K-1}+\varepsilon}^{\infty}\alpha_{\sigma(K)}\left(X_{x_{K}}\right)e^{-A_{\sigma(K)}\left(x_{K}\right)-A_{0}\left(x_{K}\right)}\,dx_{K}.
\end{multline}
Now, fix a similar permutation to $j$ with the only difference being
that $j_{K}=0$ and $j_{K+1}=\sigma(K)$. By a similar calculation
as in (\ref{eq: (2)}), we get, 
\begin{multline}
    \int\limits _{x_{K-1}+\varepsilon}^{\infty}f_{j_{K}}(x_{K})\int\limits _{x_{K}}^{\infty}f_{j_{K+1}}(x_{K+1})\,dx_{K+1}\,dx_{K}\\
    =\int\limits _{x_{K-1}+\varepsilon}^{\infty}\alpha_{0}\left(X_{x_{K}}\right)e^{-A_{0}\left(x_{K}\right)-A_{\sigma(K)}\left(x_{K}\right)}\,dx_{K}.  \label{Proof 1.3}
\end{multline}
After noting that the $K-1$ outer integrals of the 2 fixed permutations
are the same, we can join (\ref{eq: (2)}) and (\ref{Proof 1.3})
to get, 
\begin{multline} \label{eq: (4)}
    \int\limits _{0}^{\infty}f_{j_{1}}(x_{1})\int\limits _{x_{1}+\varepsilon}^{\infty}f_{j_{2}}(x_{2})\dots \\
    \dots \int\limits _{x_{K-1}+\varepsilon}^{\infty}\left[\alpha_{0}\left(X_{x_{K}}\right)+\alpha_{\sigma(K)}\left(X_{x_{K}}\right)\right]e^{-A_{0}\left(x_{K}\right)-A_{\sigma(K)}\left(x_{K}\right)}\,dx_{K}\dots dx_{1}\\
    =\int\limits _{0}^{\infty}f_{j_{1}}(x_{1})\int\limits _{x_{1}+\varepsilon}^{\infty}f_{j_{2}}(x_{2})\dots \\
    \dots \int\limits _{x_{K-2}+\varepsilon}^{\infty}f_{\sigma(K)}(x_{K-1})e^{-A_{0}\left(x_{K-1}+\varepsilon\right)-A_{\sigma(K)}\left(x_{K-1}+\varepsilon\right)}\,dx_{K-1}\dots dx_{1}.
\end{multline}
Now, we only need to take the sum of terms like (\ref{eq: (4)}) over
all the possible $K!$ permutations of $(1,2,\dots,K)$ and the result
follows by taking a further expectation. 
\end{proof}
This theorem shows that the market failure time is computable given
the banks' and the market-wide stress event's intensity processes.
In this form, it is quite abstract. To understand this probability
better, we consider two special cases in the subsequent remarks.
\begin{cor}
\emph{(Identically Distributed Default Times)}

If we assume $\alpha(X_x)=\alpha_{1}(X_x)=\alpha_{2}(X_x)=\dots=\alpha_{K}(X_x)$
a.s., i.e. $\eta_{1},\eta_{2},\dots,\eta_{K}$ are identically distributed
given $\mathcal{F}^{X}_\infty$, then in Theorem \ref{Thm Market Failure Probability}
we can replace the sum over all the permutations by $K!$. More specifically, 
\begin{multline}
    \mathbb{P}\left(\left|\tau_{i}-\tau_{j}\right|<\varepsilon\;for\:some\:(i,j)\in\left(1,...,K\right)\times\left(1,...,K\right),{\color{red}{\normalcolor i\neq j}}\right) \\
    = 1-\mathbb{E}\left[K!\int\limits _{0}^{\infty}f(x_{1})\int\limits _{x_{1}+\varepsilon}^{\infty}f(x_{2})\dots\int\limits _{x_{K-3}+\varepsilon}^{\infty}f(x_{K-2}\right.) \\
    \left.\int\limits _{x_{K-2}+\varepsilon}^{\infty}\alpha(X_{x_{K-1}})e^{-A(x_{K-1})-A\left(x_{K-1}+\varepsilon\right)-A_{0}\left(x_{K-1}+\varepsilon\right)}\,dx_{K-1}\,dx_{K-2}\dots dx_{1}\right]
\end{multline}
where $f(x)=\alpha(X_{x})e^{-A(x)}$. 
\end{cor}

\medskip{}

\begin{cor}
(\emph{Constant Default Intensities}) 

If $\alpha_i(X_t) = \alpha_i$ for all $t\geq 0$, all $i=1, 2, \dots, K$, and where $\alpha_i\in \mathbb{R}^+$, then the market failure probability is
\begin{multline}
\mathbb{P}\left(\left|\tau_{i}-\tau_{j}\right|<\varepsilon\;for\:some\:(i,j)\in\left(1,...,K\right)\times\left(1,...,K\right),i\neq j{\color{black}}\right)=\\
1-\sum_{j\in P_K}\prod\limits _{i=1}^{K-1}\left[\frac{\alpha_{j_{i}}}{\alpha_{0}+\sum\limits _{k=i}^{K}\alpha_{j_{k}}}e^{-\varepsilon\left(\alpha_{0}+\sum\limits _{k=i+1}^{K}\alpha_{j_{k}}\right)}\right]. \label{eq: mkt failure prob constant inten}
\end{multline}
Here, we see that the market failure probability is easily computed
given estimates of the default intensities.
\end{cor}

\begin{proof}
(Expression (\ref{eq: mkt failure prob constant inten})) Use Theorem
(\ref{Thm Market Failure Probability}) with $\alpha_{k}(x)\equiv\alpha_{k}$
which implies $f_{k}(x)=\alpha_{k}e^{-\alpha_{k}x}$ for all $k$.
Solve the integrals to get

\begin{multline*}
    \mathbb{P}\left(\left|\tau_{i}-\tau_{j}\right|\geq\varepsilon\text{ for all }(i,j),\:i\neq j\right)\\
    = \sum_{j\in P_K} \left[\frac{\alpha_{j_{1}}}{\alpha_{0}+\alpha_{j_{1}}+\alpha_{j_{2}}+\dots+\alpha_{j_{K}}}\right]\left[\frac{\alpha_{j_{2}}}{\alpha_{0}+\alpha_{j_{2}}+\dots+\alpha_{j_{K}}}\right]\dots \\
    \dots\left[\frac{\alpha_{j_{K-2}}}{\alpha_{0}+\alpha_{j_{K-2}}+\alpha_{j_{K-1}}+\alpha_{j_{K}}}\right]\left[\frac{\alpha_{j_{K-1}}}{\alpha_{0}+\alpha_{j_{K-1}}+\alpha_{j_{K}}}\right]\exp\left[-\varepsilon\left(\alpha_{0}+\alpha_{j_{2}}+ \right.\right.\\
    \left. \left. +\alpha_{j_{3}} \dots +\alpha_{j_{K}}\right)\right] \times\exp\left[-\varepsilon\left(\alpha_{0}+\alpha_{j_{3}}+\alpha_{j_{4}}+\dots+\alpha_{j_{K}}\right)\right] \times \dots \\
    \dots\times\exp\left[-\varepsilon\left(\alpha_{0}+\alpha_{j_{K-1}}+\alpha_{j_{K}}\right)\right]\exp\left[-\varepsilon\left(\alpha_{0}+\alpha_{j_{K}}\right)\right].
\end{multline*}
\end{proof}

\begin{example}
	The next numerical example illustrates how our measure of systemic risk (Theorem \ref{Thm Market Failure Probability}) can be implemented in practice. 
	Suppose $K=3$\footnote{We do not present the case for larger $K$ because this imposes computational challenges that are outside the scope of this paper.}, $\varepsilon=1$, and that $\alpha_i(X_t) = \alpha_i$ for all $t\geq 0$, all $i=0, 1, 2, 3$, where $\alpha_i\in \mathbb{R}^+$.
	We assume we have 3 estimates of $\alpha_i$ for $i=1,2,3$, namely $\bar{\alpha}_i, \tilde{\alpha}_i$, and $\hat{\alpha}_i$. Each one of these are random numbers between $0.00004$ and $0.00006$ which are reasonable estimates of $\alpha_i$ for any given G-SIB (See Table 7 in \cite{Duffie.Jarrow.Purnanandam.Yang}).
	
	We further assume that the intensity rate of a market failure is approximately 10 times smaller than the probability of a single bank failing (due to idiosyncratic events). Although 10 times smaller is an arbitrary selection, the intent is to make the likelihood of a market failure an order of magnitude smaller than a single bank's failure. The purpose of which is to illustrate the impact of varying this probability on our systemic risk measure. Hence, we compute the change in our measure of systemic risk when varying $\alpha_0$ between $0.000005$ and $0.00001$. The result is illustrated in figure \ref{fig:plot} where the y-axis represents the numerical value of formula (\refeq{eq: mkt failure prob constant inten}).

\begin{figure} [h!]
	\includegraphics[width=.7\textwidth]{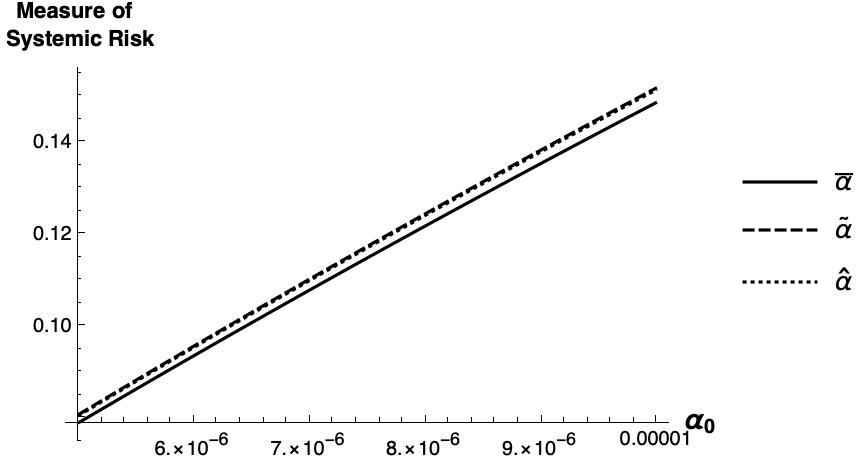}
	\centering
	\caption{Change in our measure of systemic risk for different estimates of $\alpha_i$ ($i=1,2,3$) when varying $\alpha_0$.}
	\label{fig:plot}
\end{figure}

	As expected, our measure of systemic risk increases for larger values of $\alpha_0$.
\end{example}

The following corollary related to the market failure probability
as $\varepsilon\rightarrow0$ follows easily from Theorem \ref{Theorem.Joint.Distribution}.
\begin{cor}
(Market Failure when $\varepsilon\rightarrow0$)
\end{cor}

\textcolor{black}{
\begin{equation}
\lim_{\varepsilon\rightarrow0}\frac{\mathbb{P}\left(\tau_{1}\in\left(t,t+\varepsilon\right],\tau_{2}\in\left(t,t+\varepsilon\right]|\tau_{1}>t,\tau_{2}>t,\right)}{\varepsilon}=\mathbb{E}\left(\alpha_{0}(X_{t})\right). \label{instantaneous.rate}
\end{equation}
}
\begin{proof}
Note that

\begin{align*}
\mathbb{P} & \left(\tau_{1}\in\left(t,t+\varepsilon\right],\tau_{2}\in\left(t,t+\varepsilon\right]|\tau_{1}>t,\tau_{2}>t,\left(X_{u}\right)_{u\leq t+\varepsilon}\right)\\
 & =\frac{\mathbb{P}\left(\tau_{1}\in\left(t,t+\varepsilon\right],\tau_{2}\in\left(t,t+\varepsilon\right]|\left(X_{u}\right)_{u\leq t+\varepsilon}\right)}{\mathbb{P}\left(\tau_{1}>t,\tau_{2}>t|\left(X_{u}\right)_{u\leq t+\varepsilon}\right)} .
\end{align*}
For the numerator, using Theorem \ref{Theorem.Joint.Distribution}
with $K=2$, 
\begin{multline*}
    \mathbb{P} \left(\tau_{1}\in\left(t,t+\varepsilon\right],\tau_{2}\in\left(t,t+\varepsilon\right]|\left(X_{u}\right)_{u\leq t+\varepsilon}\right)=\\
    \mathbb{P}\left(\tau_{1}>t,\tau_{2}>t|\left(X_{u}\right)_{u\leq t+\varepsilon}\right) -\mathbb{P}\left(\tau_{1}>t,\tau_{2}>t+\varepsilon|\left(X_{u}\right)_{u\leq t+\varepsilon}\right)\\
    -\mathbb{P}\left(\tau_{1}>t+\varepsilon,\tau_{2}>t|\left(X_{u}\right)_{u\leq t+\varepsilon}\right)    +\mathbb{P}\left(\tau_{1}>t+\varepsilon,\tau_{2}>t+\varepsilon|\left(X_{u}\right)_{u\leq t+\varepsilon}\right)= \\
    \exp\left[-\left(A_{1}+A_2+A_0\right)\left(t\right)\right]- \exp\left[-A_{1}\left(t\right)-\left(A_{2}+A_0\right)\left(t+\varepsilon\right)\right]\\
    -\exp\left[-\left(A_{1}+A_0\right)\left(t+\varepsilon\right)-A_{2}\left(t\right)\right] +\exp\left[-\left(A_{1}+A_2 +A_0 \right)\left(t+\varepsilon\right)\right].
\end{multline*}
Dividing by $\mathbb{P}\left(\tau_{1}>t,\tau_{2}>t|\left(X_{u}\right)_{u\leq t+\varepsilon}\right)=\exp\left[-A_{1}\left(t\right)-A_{2}\left(t\right)-A_{0}\left(t\right)\right]$,
we get: 
\begin{align*}
\mathbb{P} & \left(\tau_{1}\in\left(t,t+\varepsilon\right],\tau_{2}\in\left(t,t+\varepsilon\right]|\tau_{1}>t,\tau_{2}>t,\left(X_{u}\right)_{u\leq t+\varepsilon}\right)=1\\
 & -\exp\left[-\int\limits _{t}^{t+\varepsilon}\left(\alpha_{2}+\alpha_{0}\right)(X_{u})du\right]-\exp\left[-\int\limits _{t}^{t+\varepsilon}\left(\alpha_{1}+\alpha_{0}\right)(X_{u})du\right]\\
 & +\exp\left[-\int\limits _{t}^{t+\varepsilon}\left(\alpha_{1}+\alpha_{2}+\alpha_{0}\right)(X_{u})du\right].
\end{align*}
Hence, by using L'H\^opital's rule, we get: 
\begin{multline*}
    \lim_{\varepsilon\rightarrow0} \frac{\mathbb{P}\left(\tau_{1}\in\left(t,t+\varepsilon\right],\tau_{2}\in\left(t,t+\varepsilon\right]|\tau_{1}>t,\tau_{2}>t,\left(X_{u}\right)_{u\leq t+\varepsilon}\right)}{\varepsilon}=\\
    \left(\alpha_{2}+\alpha_{0}\right)\left(X_{t}\right) +\left(\alpha_{1}+\alpha_{0}\right)\left(X_{t}\right)-\left(\alpha_{1}+\alpha_{2}+\alpha_{0}\right)\left(X_{t}\right)=\alpha_{0}(X_{t}).
\end{multline*}
Moreover, given $\mathbb{P}\left(\tau_{1}\in\left(t,t+\varepsilon\right],\tau_{2}\in\left(t,t+\varepsilon\right]|\tau_{1}>t,\tau_{2}>t,\left(X_{u}\right)_{u\leq t+\varepsilon}\right)$
is bounded by 1 and a conditioning argument, we can conclude that
\begin{equation}
\lim_{\varepsilon\rightarrow0}\frac{\mathbb{P}\left(\tau_{1}\in\left(t,t+\varepsilon\right],\tau_{2}\in\left(t,t+\varepsilon\right]|\tau_{1}>t,\tau_{2}>t,\right)}{\varepsilon}=\mathbb{E}\left(\alpha_{0}(X_{t})\right).
\end{equation}
\end{proof}
\textcolor{black}{That is, the probability of a market failure due
to two banks defaulting over the next time interval $(t,t+\varepsilon]$
with $\varepsilon\approx0$ is equal to the probability of a market-wide
stress event occurring, as noted before.}

\subsection{Catastrophic Market Failure}

It is of interest to determine the probability that all banks will
default within a small interval of time. This would be a catastrophic
market failure.
\begin{thm}
(Occurrence of all the Default Times)

Let \textup{$\tau_{(K)}=max\{\tau_{1},...,\tau_{K}\}$.} 
\begin{equation}
\mathbb{P}\left(\tau_{(K)}\leq\varepsilon\right)=1+\mathbb{E}\left[e^{-A_{0}\left(\varepsilon\right)}\left(\prod\limits _{i=1}^{K}\left(1-e^{-A_{i}\left(\varepsilon\right)}\right)-1\right)\right].
\end{equation}
\end{thm}

\begin{proof}
For any $K$, we have 
\begin{align*}
\tilde{\mathbb{P}}\left(\tau_{(K)}\leq\varepsilon\right)= & \tilde{\mathbb{P}}\left(\tau_{1}\leq\varepsilon,\tau_{2}\leq\varepsilon,\dots,\tau_{K}\leq\varepsilon\right)\\
= & \tilde{\mathbb{P}}\left(\tau_{1}\leq\varepsilon,\tau_{2}\leq\varepsilon,\dots,\tau_{K}\leq\varepsilon,\eta_{0}>\varepsilon\right)\\
 & +\tilde{\mathbb{P}}\left(\tau_{1}\leq\varepsilon,\tau_{2}\leq\varepsilon,\dots,\tau_{K}\leq\varepsilon,\eta_{0}\leq\varepsilon\right).\\
\intertext{As \ensuremath{\tau_{i}=\min\left(\eta_{i},\eta_{0}\right)}, the event \ensuremath{\left\{ \tau_{1}\leq\varepsilon,\tau_{2}\leq\varepsilon,\dots,\tau_{K}\leq\varepsilon,\eta_{0}>\varepsilon\right\} } is equal to \ensuremath{\left\{ \eta_{1}\leq\varepsilon,\eta_{2}\leq\varepsilon,\dots,\eta_{K}\leq\varepsilon,\eta_{0}>\varepsilon\right\} } and \ensuremath{\left\{ \tau_{1}\leq\varepsilon,\tau_{2}\leq\varepsilon,\dots,\tau_{K}\leq\varepsilon,\eta_{0}\leq\varepsilon\right\} } is equal to \ensuremath{\left\{ \eta_{0}\leq\epsilon\right\} }}
\tilde{\mathbb{P}}(\tau_{(K)}\leq \varepsilon) = & \tilde{\mathbb{P}}\left(\eta_{1}\leq\varepsilon,\eta_{2}\leq\varepsilon,\dots,\eta_{K}\leq\varepsilon,\eta_{0}>\varepsilon\right)+\tilde{\mathbb{P}}\left(\eta_{0}\leq\varepsilon\right). \\
\intertext{Using the independence of \ensuremath{\eta_{i}}}
\tilde{\mathbb{P}}(\tau_{(K)}\leq \varepsilon) = & \tilde{\mathbb{P}}\left(\eta_{0}>\varepsilon\right)\tilde{\mathbb{P}}\left(\eta_{1}\leq\varepsilon\right)\tilde{\mathbb{P}}\left(\eta_{2}\leq\varepsilon\right)\dots\tilde{\mathbb{P}}\left(\eta_{K}\leq\varepsilon\right)+\tilde{\mathbb{P}}\left(\eta_{0}\leq\varepsilon\right)\\
= & e^{-A_{0}\left(\varepsilon\right)}\left[\prod\limits _{i=1}^{K}\left(1-e^{-A_{i}(\varepsilon)}\right)\right]+1-e^{-A_{0}\left(\varepsilon\right)}\\
= & 1+e^{-A_{0}\left(\varepsilon\right)}\left[\prod\limits _{i=1}^{K}\left(1-e^{-A_{i}(\varepsilon)}\right)-1\right].
\end{align*}
The result follows by taking an additional expectation.
\end{proof}
\medskip{}

\begin{rem}
(\emph{Catastrophic Market Failure with Identical Distributions}) 
\end{rem}

Let $A_{i}(\varepsilon)$ for $i=1,2,\dots,K$ be the same for all
banks and independent of $K$, say $A(\varepsilon)$, implying that
the G-SIBs' default times are identically distributed. Then, $\lim_{K\rightarrow\infty}\prod\limits _{i=1}^{K}\left(1-e^{-A_i\left(\varepsilon\right)}\right)= \lim_{K\rightarrow\infty} \left(1-e^{-A \left(\varepsilon\right)}\right)^K= 0$,
which implies that
\begin{equation}
\lim_{K\rightarrow\infty}\mathbb{P}\left(\tau_{(K)}<\varepsilon\right)=1-\mathbb{E}\left(e^{-A_{0}(\varepsilon)}\right)=\mathbb{P}\left(\eta_{0}<\varepsilon\right).
\end{equation}
As the number of G-SIBs approaches infinity, the probability of a
catastrophic market failure is equal to the probability of a market-wide
stress event occurring. This makes sense since, except for the market-wide
stress event, the probability of idiosyncratic defaults are independent
across G-SIBs.

\medskip{}

\begin{cor}
(\emph{Catastrophic Market Failure with Constant Default Intensities}) \par
If $\alpha_i(X_t) = \alpha_i$ for all $t\geq 0$, all $i=1, 2, \dots, K$, and where $\alpha_i\in \mathbb{R}^+$, then:
\begin{equation}
\mathbb{P}\left(\tau_{(K)}\leq\varepsilon\right)=1+e^{-\alpha_{0}\varepsilon}\left(\prod\limits _{i=1}^{K}\left(1-e^{\alpha_{i}\varepsilon}\right)-1\right).
\end{equation}
\end{cor}

\medskip{}

\begin{rem}
(\emph{Destructive Competition})

With destructive competition (see Example \ref{exa:(Destructive-Competition)})
and assuming that for all $i=1,2,\dots,K$, $\alpha_{i}(X_{t},K)=\ln(K)+b(X_{t})$ and $\varepsilon<1$, then
\begin{equation}
\lim_{K\rightarrow\infty}\mathbb{P}\left(\tau_{(K)}<\varepsilon\right)=\mathbb{P}\left(\eta_{0}<\varepsilon\right).
\end{equation}

As documented, under destructive competition, the probability of a catastrophic market failure increases to the indicated limit as the number of GSIBs approaches infinity.
\end{rem}

\begin{proof}
(Destructive Competition) Recall that:
\begin{equation}
\mathbb{P}\left(\tau_{(K)}\leq\varepsilon\right)=1-\mathbb{E}\left[e^{-A_{0}\left(\varepsilon\right)}\left(1-\prod\limits _{i=1}^{K}\left(1-e^{-A_{i}\left(\varepsilon,K\right)}\right)\right)\right].
\end{equation}
If $\alpha_{i}(X_{t},K)=\ln(K)+b(X_{t})$, then 
\begin{align}
\tilde{\mathbb{P}}\left(\tau_{(K)}\leq\varepsilon\right) & =1-e^{-A_{0}\left(\varepsilon\right)}\left(1-\left(1-e^{-A\left(\varepsilon,K\right)}\right)^{K}\right)\nonumber \\
 & =1-e^{-A_{0}\left(\varepsilon\right)}\left(1-\left(1-e^{-\varepsilon\ln(K)-B(\varepsilon)}\right)^{K}\right)\nonumber \\
 & =1-e^{-A_{0}\left(\varepsilon\right)}\left(1-\left(1-\frac{e^{-B(\varepsilon)}}{K^{\varepsilon}}\right)^{K}\right)  \nonumber \\
 & \overset{K\rightarrow\infty}{\longrightarrow}1-e^{-A_{0}\left(\varepsilon\right)}.
\end{align}
The results follow by taking an expectation and interchanging limits,
which we can do because $\tilde{\mathbb{P}}(\cdot)$ is bounded by
1. 
\end{proof}

\subsection{Bounds on the Market Failure Probability}

For some empirical applications, it may be useful to obtain bounds
on the market failure probability as in the subsequent theorem. 
\begin{thm}
(Bounds on the Market Failure Probability)
\begin{multline}
    \mathbb{P}\left(\left|\tau_{i}-\tau_{j}\right|<\varepsilon\text{ for some }(i,j), i\neq j\right)\leq \\ 
    \min \left [ \binom{K}{2}-\sum\limits _{i=1}^{K}\sum\limits _{j\neq i}^{K}\mathbb{E}\left( \int_{0}^{\infty}\alpha_{i}(X_{x})e^{-A_{i}(x)-A_{j}(x+\varepsilon)-A_{0}(x+\varepsilon)}dx\right),  \right. \\
    1-\mathbb{E}\Biggl(\sum_{j\in S(P_K)}\int\limits _{0}^{\infty}f_{j_{1}}(x_{1})\int\limits _{x_{1}+\varepsilon}^{\infty}f_{j_{2}}(x_{2})\int\limits _{x_{2}+\varepsilon}^{\infty}f_{j_{3}}(x_{3})\dots\int\limits _{x_{K-3}+\varepsilon}^{\infty}f_{j_{K-2}}(x_{K-2}) \\
    \int\limits _{x_{K-2}+\varepsilon}^{\infty}f_{j_{K-1}}(x_{K-1})\exp\left[-A_{j_{K}}(x_{K-1}+\varepsilon)-A_{0}(x_{K-1}+\varepsilon)\right]\,dx_{K-1}\,dx_{K-2}\dots\\
    \left.  \dots\,dx_{3}\,dx_{2}\,dx_{1}\Biggr) \right]
\end{multline}
where $S\left(P_K\right)$ stands for a subset of $P_K$ (Recall $P_K$ stands for the set of all possible permutations of $(1,2,\dots, K)$).

\begin{multline}
    \mathbb{P}\left(\left|\tau_{i}-\tau_{j}\right|<\varepsilon\text{ for some }(i,j), i \neq j\right)\geq \\ 
    \max \left[ 1-\mathbb{E}\left( e^{-A_{0}\left((K-1)\varepsilon\right)}\sum\limits _{j\in P_K}\exp\left(-\sum\limits _{i=1}^{K-1}A_{j_{i+1}}\left(i\varepsilon\right)\right)\right), \right. \\
    \left. 1 - \min\limits_{i \neq j} \left( \int_0^\infty \alpha_i(x) e^{-A_i(x)-(A_j+A_0)(x+\varepsilon)}dx +  \int_0^\infty \alpha_j(x) e^{-A_j(x)-(A_i+A_0)(x+\varepsilon)}dx \right)\right].
\end{multline}
\end{thm}

\begin{proof}
Note that by construction $\eta_{0},\eta_{1},\eta_{2},\dots,\eta_{K}$
are mutually independent. 

For the first upper bound, note the following: 
\begin{align}\label{Equation.Star}
\tilde{\mathbb{P}} & \left(\left|\tau_{i}-\tau_{j}\right|<\varepsilon\;for\:some\:(i,j)\in\left(1,...,K\right)\times\left(1,...,K\right),i\neq j\right) \nonumber \\
 & =\tilde{\mathbb{P}}\left(\bigcup_{i=1}^{K}\bigcup_{j>i}|\tau_{i}-\tau_{j}|\leq\varepsilon\right)\leq\sum\limits _{i=1}^{K}\sum\limits _{j>i}\tilde{\mathbb{P}}\left(|\tau_{i}-\tau_{j}|\leq\varepsilon\right) .
\end{align}
For any $i,j$, we have: 
\[
\tilde{\mathbb{P}}\left(|\tau_{i}-\tau_{j}|\leq\varepsilon\right)=1-\tilde{\mathbb{P}}\left(\tau_{i}-\tau_{j}>\varepsilon,\tau_{i}>\tau_{j}\right)-\tilde{\mathbb{P}}\left(\tau_{j}-\tau_{i}>\varepsilon,\tau_{j}>\tau_{i}\right).
\]
Let us focus on $\tilde{\mathbb{P}}\left(\tau_{i}-\tau_{j}>\varepsilon,\tau_{i}>\tau_{j}\right)$
and note that in the next equality we do not consider the event $\{\eta_{i}>\eta_{j}>\eta_{0}\}$
because it implies $\tau_{i}=\tau_{j}$ and so it is impossible to
have $\tau_{i}-\tau_{j}>\varepsilon$ 
\begin{align*}
    \tilde{\mathbb{P}} & \left(\tau_{i}-\tau_{j}>\varepsilon,\tau_{i}>\tau_{j}\right)=\tilde{\mathbb{P}}\left(\tau_{i}-\tau_{j}>\varepsilon,\eta_{0}>\eta_{i}>\eta_{j}\right)\\
    & +\tilde{\mathbb{P}}\left(\tau_{i}-\tau_{j}>\varepsilon,\eta_{i}>\eta_{0}>\eta_{j}\right) =\int\limits _{0}^{\infty}\int\limits _{y+\varepsilon}^{\infty}\int\limits _{x}^{\infty}f_{0}(z)f_{i}(x)f_{j}(y)\,dz\,dx\,dy\\
    & +\int\limits _{0}^{\infty}\int\limits _{y+\varepsilon}^{\infty}\int\limits _{z}^{\infty}f_{i}(x)f_{0}(z)f_{j}(y)\,dx\,dz\,dy =\int\limits _{0}^{\infty}f_{j}(y)\int\limits _{y+\varepsilon}^{\infty}\alpha_{i}(X_{x})e^{-A_{i}(x)-A_{0}(x)}\,dx\,dy\\
    & +\int\limits _{0}^{\infty}f_{j}(y)\int\limits_{y+\varepsilon}^{\infty}\alpha_{0}(X_{z})e^{-A_{i}(z)-A_{0}(z)}\,dz\,dy\\
    & =\int\limits _{0}^{\infty}f_{j}(y)\int\limits _{y+\varepsilon}^{\infty}\left[\alpha_{i}(X_{x})+\alpha_{0}(X_{x})\right]e^{-A_{i}(x)-A_{0}(x)}\,dx\,dy \\
    & =\int\limits _{0}^{\infty}\alpha_{j}(y)e^{-A_{j}(y)-A_{i}(y+\varepsilon)-A_{0}(y+\varepsilon)}\,dy.
\end{align*}
$\tilde{\mathbb{P}}\left(\tau_{i}-\tau_{j}>\varepsilon,\tau_{j}>\tau_{i}\right)$
follows by an analogous argument: 
\[
\tilde{\mathbb{P}}\left(\tau_{i}-\tau_{j}>\varepsilon,\tau_{j}>\tau_{i}\right)=\int\limits _{0}^{\infty}\alpha_{i}(y)e^{-A_{i}(y)-A_{j}(y+\varepsilon)-A_{0}(y+\varepsilon)}\,dy.
\]
Joining these 2 probabilities, we get: 
\begin{multline*}
    \tilde{\mathbb{P}}\left(|\tau_{i}-\tau_{j}|\leq\varepsilon\right)=1-\int\limits _{0}^{\infty}\alpha_{j}(y)e^{-A_{j}(y)-A_{i}(y+\varepsilon)-A_{0}(y+\varepsilon)}dy\\
    -\int\limits _{0}^{\infty}\alpha_{i}(y)e^{-A_{i}(y)-A_{j}(y+\varepsilon)-A_{0}(y+\varepsilon)}dy.
\end{multline*}
Using this along with \eqref{Equation.Star}, we get: 
\begin{multline*}
    \tilde{\mathbb{P}} \left(\left|\tau_{i}-\tau_{j}\right|<\varepsilon\;for\:some\:(i,j)\in\left(1,...,K\right)\times\left(1,...,K\right),i\neq j\right)\\
    \leq\sum\limits _{i=1}^{K}\sum\limits _{j>i}\left[1-\int\limits _{0}^{\infty}\alpha_{j}(y)e^{-A_{j}(y)-A_{i}(y+\varepsilon)-A_{0}(y+\varepsilon)}dy \right.\\
    \left.-\int\limits _{0}^{\infty}\alpha_{i}(y)e^{-A_{i}(y)-A_{j}(y+\varepsilon)-A_{0}(y+\varepsilon)}dy\right]\\
    =\binom{K}{2}-\sum\limits _{i=1}^{K}\sum\limits _{j\neq i}^{K}\int\limits _{0}^{\infty}\alpha_{i}(y)e^{-A_{i}(y)-A_{j}(y+\varepsilon)-A_{0}(y+\varepsilon)}dy.
\end{multline*}

Finally, the desired upper bound follows by taking an additional expectation.

The second upper bound is natural as we are taking a sum over a smaller set.

For the first lower bound, it suffices to bound
\begin{equation*}
	\mathbb{P}\left(|\tau_{i}-\tau_{j}|>\varepsilon\text{ for all }(i,j), i \neq j \right)
\end{equation*}
from above. This is because: 
\begin{align*}
\mathbb{P}\left(\left|\tau_{i}-\tau_{j}\right|<\varepsilon\text{ for some }(i,j), i\neq j \right)=1-\mathbb{P}\left(|\tau_{i}-\tau_{j}|>\varepsilon\text{ for all }(i,j), i\neq j \right).
\end{align*}
From Theorem \ref{Thm Market Failure Probability}, we know that,
\begin{multline*}
    \tilde{\mathbb{P}}\left(|\tau_{i}-\tau_{j}|>\varepsilon\text{ for all }(i,j), i\neq j \right)=\sum_{j\in P_K}\int\limits _{0}^{\infty}f_{j_{1}}(x_{1})\int\limits _{x_{1}+\varepsilon}^{\infty}f_{j_{2}}(x_{2})\\
    \int\limits _{x_{2}+\varepsilon}^{\infty}f_{j_{3}}(x_{3}) \dots \int\limits _{x_{K-3}+\varepsilon}^{\infty}f_{j_{K-2}}(x_{K-2}) \int\limits _{x_{K-2}+\varepsilon}^{\infty}f_{j_{K-1}}(x_{K-1})\exp\left[-A_{j_{K}}(x_{K-1}+\varepsilon)\right.\\
    \left.-A_{0}(x_{K-1}+\varepsilon)\right]\,dx_{K-1}\,dx_{K-2}\dots\,dx_{3}\,dx_{2}\,dx_{1}.
\end{multline*}
As $x_{K-1}\geq x_{K-2}+\varepsilon$ in the innermost integral in
one of the terms of the previous sum, then $e^{-\left(A_{j_{K}}+A_{0}\right)(x_{K-1}+\varepsilon)}\leq e^{-\left(A_{j_{K}}+A_{0}\right)(x_{K-2}+2\varepsilon)}$
and hence: 
\begin{multline*}
    \int\limits _{0}^{\infty}f_{j_{1}}(x_{1})\int\limits _{x_{1}+\varepsilon}^{\infty}f_{j_{2}}(x_{2})\dots\\
    \dots \int\limits _{x_{K-2}+\varepsilon}^{\infty}f_{j_{K-1}}(x_{K-1})e^{-\left(A_{j_{K}}+A_{0}\right)(x_{K-1}+\varepsilon)} dx_{K-1}\dots\,dx_{1}\\
    \leq\int\limits _{0}^{\infty}f_{j_{1}}(x_{1})\int\limits _{x_{1}+\varepsilon}^{\infty}f_{j_{2}}(x_{2})\dots e^{-\left(A_{j_{K}}+A_{0}\right)(x_{K-2}+2\varepsilon)}\\
    \int\limits _{x_{K-2}+\varepsilon}^{\infty}f_{j_{K-1}}(x_{K-1})\,dx_{K-1}\dots\,dx_{2}\,dx_{1}\\
    =\int\limits _{0}^{\infty}f_{j_{1}}(x_{1})\int\limits _{x_{1}+\varepsilon}^{\infty}f_{j_{2}}(x_{2})\dots\\
    \dots \int\limits _{x_{K-3}+\varepsilon}^{\infty}f_{j_{K-2}}(x_{K-2})e^{-\left(A_{j_{K}}+A_{0}\right)(x_{K-2}+2\varepsilon)-A_{j_{K-1}}(x_{K-2}+\varepsilon)}\,dx_{K-2}\dots\,dx_{2}\,dx_{1}.
\end{multline*}
As $x_{K-2}\geq x_{K-3}+\varepsilon$ in the innermost integral, we
have
\begin{equation*}
    e^{-\left(A_{j_{K}}+A_{0}\right)(x_{K-2}+2\varepsilon)-A_{j_{K-1}}(x_{K-2}+\varepsilon)}\leq e^{-\left(A_{j_{K}}+A_{0}\right)(x_{K-3}+3\varepsilon)-A_{j_{K-1}}(x_{K-3}+2\varepsilon)}.
\end{equation*}
Then, we solve the integral with respect to $x_{K-2}$. We keep going
in a similar fashion to conclude that: 
\begin{multline*}
     \int\limits _{0}^{\infty}f_{j_{1}}(x_{1})\int\limits _{x_{1}+\varepsilon}^{\infty}f_{j_{2}}(x_{2})\dots\\
     \dots \int\limits _{x_{K-2}+\varepsilon}^{\infty}f_{j_{K-1}}(x_{K-1})e^{-\left(A_{j_{K}}+A_{0}\right)(x_{K-1}+\varepsilon)}\,dx_{K-1}\dots\,dx_{2}\,dx_{1}\\
    \leq \exp\left[-\left(A_{j_{K}}+A_{0}\right)\left((K-1)\varepsilon\right)-A_{j_{K-1}}\left((K-2)\varepsilon\right)-\dots \right.\\
    \left. \dots -A_{j_{3}}(2\varepsilon)-A_{j_{2}}(\varepsilon)\right]= \exp\left[-A_{0}\left((K-1)\varepsilon\right)-\sum\limits _{i=1}^{K-1}A_{j_{i+1}}\left(i\varepsilon\right)\right].
\end{multline*}
Then we take the sum over all possible permutations $P_K$ to get:
\begin{multline*}
    \tilde{\mathbb{P}}\left(|\tau_{i}-\tau_{j}|>\varepsilon\text{ for all }(i,j), i\neq j \right)\leq \\
    \sum_{j\in P_K}\exp\left[-A_{0}\left((K-1)\varepsilon\right)-\sum\limits _{i=1}^{K-1}A_{j_{i+1}}\left(i\varepsilon\right)\right].
\end{multline*}
Recall that the sub index $j$ in the sum stands for any of the possible
permutations. This is $j=\left(j_{1}=\sigma(1),j_{2}=\sigma(2),\dots,j_{K}=\sigma(K)\right)$
Using the complement of the previous probability, we have 
\begin{multline*}
\tilde{\mathbb{P}}\left(\left|\tau_{i}-\tau_{j}\right|<\varepsilon\text{ for some }(i,j), i\neq j\right)\geq \\
1-e^{-A_{0}\left((K-1)\varepsilon\right)}\sum\limits _{j\in P_K}\exp\left(-\sum\limits _{i=1}^{K-1}A_{j_{i+1}}\left(i\varepsilon\right)\right).
\end{multline*}
The desired lower bound follows by taking an additional expectation. 

For the second lower bound, note that:
	\begin{multline} \label{Lower.Bound}
		\mathbb{P}\left(|\tau_{i}-\tau_{j}|>\varepsilon\text{ for all }(i,j), i\neq j \right) = \\
		\mathbb{P}\left( \bigcap\limits_{i=1}^K \bigcap\limits_{j \neq i}^K |\tau_{i}-\tau_{j}|>\varepsilon  \right) \leq \min\limits_{i \neq j} \mathbb{P}\left(  |\tau_{i}-\tau_{j}|>\varepsilon  \right)\\
		= \min\limits_{i \neq j} \left[ \int_0^\infty \alpha_i(x) e^{-A_i(x)-(A_j+A_0)(x+\varepsilon)}dx +  \int_0^\infty \alpha_j(x) e^{-A_j(x)-(A_i+A_0)(x+\varepsilon)}dx \right].
	\end{multline}
	And then, use the following:
	\begin{equation*}
		\mathbb{P}\left(\left|\tau_{i}-\tau_{j}\right|<\varepsilon\text{ for some }(i,j), i\neq j \right)=1-\mathbb{P}\left(|\tau_{i}-\tau_{j}|>\varepsilon\text{ for all }(i,j), i\neq j\right).
	\end{equation*}
\end{proof}
\begin{cor}
(\emph{Constant Default Intensities}) \par
If $\alpha_i(X_t) = \alpha_i$ for all $t\geq 0$, all $i=1, 2, \dots, K$, and where $\alpha_i\in \mathbb{R}^+$, then:
\begin{multline}
    \mathbb{P}\left(\left|\tau_{i}-\tau_{j}\right|<\varepsilon\text{ for some }(i,j), i\neq j\right)\leq \\ 
    \min \left[ \binom{K}{2}-e^{-\alpha_{o}\varepsilon}\sum\limits _{i=1}^{K}\alpha_{i}\sum\limits _{j\neq i}^{K}e^{\alpha_{j}\varepsilon}\left(\frac{1}{\alpha_{i}+\alpha_{j}+\alpha_{0}}\right), \right. \\
    \left. 	1-\sum_{j\in S\left( P_K \right)}\prod\limits _{i=1}^{K-1}\left(\frac{\alpha_{j_{i}}}{\alpha_{0}+\sum\limits _{k=i}^{K}\alpha_{j_{k}}}e^{-\varepsilon\left(\alpha_{0}+\sum\limits _{k=i+1}^{K}\alpha_{j_{k}}\right)}\right) \right].
\end{multline}
\begin{multline}
    \mathbb{P}\left(\left|\tau_{i}-\tau_{j}\right|<\varepsilon\text{ for some }(i,j), i\neq j \right)\geq \\ 
    \max \left[ 1-e^{\alpha_{0}\varepsilon\left(K-1\right)}\sum\limits _{j\in P_K}\exp\left(-\varepsilon\sum\limits _{i=1}^{K-1}i\alpha_{j_{i+1}}\right), \right. \\
     \left. 1- \min_{i \neq j} \left( e^{-\varepsilon \alpha_0} \frac{1}{\alpha_i + \alpha_j + \alpha_0} \left(\alpha_i e^{-\varepsilon \alpha_j} + \alpha_j e^{-\varepsilon \alpha_i}\right)  \right) \right] .
\end{multline}
\end{cor}

\section{Comparative Statics}

This section explores how the market failure probability changes when
the initial conditions in the economy change. We consider both changing
the number of G-SIBs and the initial state of the economy.

\subsection{The Number of G-SIBs}

For regulatory purposes, it is important to understand how the probability
of a market failure changes with the inclusion of another G-SIB. This
relates to macro-prudential policy regarding whether the number of
banks in the economy being ``too large to fail'' or designated as
G-SIBs should be reduced (by breaking them up into smaller institutions)
to decrease the market failure probability (see \cite{Berndt.Duffie.Zhu.2021}, 
\cite{Schich.Toader.2017} for issues related to G-SIBs designation).

Computing the market failure probability for $K$ banks versus $K+1$
banks in Theorem \ref{Thm Market Failure Probability} and taking
the difference yields the marginal impact of adding another G-SIB
to the economy. It is easy to show that this probability increases
as more G-SIBs enter the market.
\begin{thm}
(Increasing the Number of G-SIBs)
\begin{multline}
    \mathbb{P}\left(\left|\tau_{i}-\tau_{j}\right|<\varepsilon\;for\:some\:(i,j)\in\left(1,\dots,K\right)\times\left(1,\dots,K\right),i\neq j\right)<\\
    \mathbb{P}\left(\left|\tau_{i}-\tau_{j}\right|<\varepsilon\;for\:some\:(i,j)\in\left(1,\dots,K+1\right)\times\left(1, \dots,K+1\right),i\neq j\right) \label{eq: result-1}
\end{multline}
provided that $\alpha_{i}(\cdot)$ for $i=0,1,2,\dots,K$ and the
underlying process, i.e., $(X_{t})_{t\geq0}$ remain fixed. 
\end{thm}
\begin{proof}
Let
\begin{align}
    A & :=\{\left|\tau_{i}-\tau_{j}\right|<\varepsilon\;for\:some\:(i,j)\in\left(1,...,K\right)^2,i\neq j\}\\
    B & :=\{\left|\tau_{i}-\tau_{j}\right|<\varepsilon\;for\:some\:(i,j)\in\left(1,\dots,K,K+1\right)^2,i\neq j\}.
\end{align}
Note that the event $B$ can be decomposed in the following way: 
\begin{equation}
B=A\cup\{\left|\tau_{K+1}-\tau_{j}\right|<\varepsilon\;for\:some\:j\in\left(1,...,K\right)\}.
\end{equation}
It is clear that $A \subsetneq B$ and hence the result follows .
\end{proof}
This result suggests that as the number of G-SIBs goes to infinity,
the market failure probability converges to one, as the following
theorem documents.
\begin{thm}
(Limit as $K\rightarrow\infty$)
\begin{equation}
\lim_{K\rightarrow\infty}\mathbb{P}\left(\left|\tau_{i}-\tau_{j}\right|<\varepsilon\;for\:some\:(i,j)\in\left(1,...,K\right)\times\left(1,...,K\right),i\neq j\right)=1.
\end{equation}
\end{thm}

\begin{proof}
Similar to the proof of Theorem 3, for $k=1,2,\dots,K$, let $\tau_{(k)}$
and $\eta_{(k)}$ be the $k^{th}$ order statistic of $\left(\tau_{1},\tau_{2},\dots,\tau_{K}\right)$
and $\left(\eta_{1},\eta_{2},\dots,\eta_{K}\right)$ respectively.
For example, $\tau_{(1)}=\min\left(\tau_{1},\tau_{2},\dots,\tau_{K}\right)$,
$\eta_{(1)}=\min\left(\eta_{1},\eta_{2},\dots,\eta_{K}\right)$, $\tau_{(K)}=\max\left(\tau_{1},\tau_{2},\dots,\tau_{K}\right)$
and $\eta_{(K)}=\max\left(\eta_{1},\eta_{2},\dots,\eta_{K}\right)$.

The complement of the event $\left\{ \left|\tau_{i}-\tau_{j}\right|<\varepsilon\;\mathrm{for\:some}\:(i,j),\:i\neq j\right\} $ is equal to the event $\left\{ \left|\tau_{i}-\tau_{j}\right|\geq\varepsilon\;\mathrm{for\:all}\:(i,j),\:i\neq j\right\} $. Then, we get: 
\begin{multline}
    \tilde{\mathbb{P}}\left(\left|\tau_{i}-\tau_{j}\right|\geq\varepsilon\;\mathrm{for\:all}\:(i,j),\:i\neq j\right)= \\
    \tilde{\mathbb{P}}\left(\left|\tau_{i}-\tau_{j}\right|\geq\varepsilon\;\mathrm{for\:all}\:(i,j),\:i\neq j,\eta_{0}\geq\eta_{(K)}\right)\\
    +\tilde{\mathbb{P}}\left(\left|\tau_{i}-\tau_{j}\right|\geq\varepsilon\;\mathrm{for\:all}\:(i,j),\:i\neq j,\eta_{(k-1)}\leq\eta_{0}<\eta_{(K)}\right)\\
    +\tilde{\mathbb{P}}\left(\left|\tau_{i}-\tau_{j}\right|\geq\varepsilon\;\mathrm{for\:all}\:(i,j),\:i\neq j,\eta_{0}<\eta_{(K-1)}\right). \label{Proof.Prop.1}
\end{multline}
If $\eta_{0}<\eta_{(K-1)}$, then at least there exists one pair of
$(i,j)$, $i\neq j$ such that $\tau_{i}=\tau_{j}=\eta_{0}$ and hence,
$\tilde{\mathbb{P}}\left(\left|\tau_{i}-\tau_{j}\right|\geq\varepsilon\;\mathrm{for\:all}\:(i,j),\:i\neq j,\eta_{0}<\eta_{(K-1)}\right)=0$.

If $\left|\tau_{i}-\tau_{j}\right|\geq\varepsilon\;\mathrm{for\:all}\:(i,j),\:i\neq j$,
then $\tau_{(2)}\geq\tau_{(1)}+\varepsilon$, $\tau_{(3)}\geq\tau_{(2)}+\varepsilon,\dots,\tau_{(K)}\geq\tau_{(K-1)}+\varepsilon$,
which implies $\tau_{(K)}\geq\tau_{(1)}+\left(K-1\right)\varepsilon$.
Moreover, on the event $\{\eta_{0}\geq\eta_{(K)}\}$, we have that
$\tau_{(K)}=\eta_{(K)}$ and $\tau_{(1)}=\eta_{(1)}$. Hence, 
\begin{multline}
    \tilde{\mathbb{P}}\left(\left|\tau_{i}-\tau_{j}\right|\geq\varepsilon\;\mathrm{for\:all}\:(i,j),\:i\neq j,\eta_{0}\geq\eta_{(K)}\right) \\
    \leq \tilde{\mathbb{P}}\left(\tau_{(K)}\geq\tau_{(1)}+\left(K-1\right)\varepsilon,\eta_{0}\geq\eta_{(K)}\right) \\
    =\tilde{\mathbb{P}}\left(\eta_{(K)}\geq\eta_{(1)}+\left(K-1\right)\varepsilon,\eta_{0}\geq\eta_{(K)}\right) \\
    =\tilde{\mathbb{P}}\left(\eta_{(1)}+\left(K-1\right)\varepsilon\leq\eta_{(K)}\leq\eta_{0}\right) \\
    \leq\tilde{\mathbb{P}}\left(\eta_{(1)}+\left(K-1\right)\varepsilon\leq\eta_{0}\right).
\end{multline}
\begin{multline}
    \tilde{\mathbb{P}}\left(\left|\tau_{i}-\tau_{j}\right|\geq\varepsilon\;\mathrm{for\:all}\:(i,j),\:i\neq j,\eta_{0}\geq\eta_{(K)}\right) \\
    \leq \tilde{\mathbb{P}}\left(\tau_{(K)}\geq\tau_{(1)}+\left(K-1\right)\varepsilon,\eta_{0}\geq\eta_{(K)}\right) \\
    =\tilde{\mathbb{P}}\left(\eta_{(K)}\geq\eta_{(1)}+\left(K-1\right)\varepsilon,\eta_{0}\geq\eta_{(K)}\right) \\
    =\tilde{\mathbb{P}}\left(\eta_{(1)}+\left(K-1\right)\varepsilon\leq\eta_{(K)}\leq\eta_{0}\right) \\
    \leq\tilde{\mathbb{P}}\left(\eta_{(1)}+\left(K-1\right)\varepsilon\leq\eta_{0}\right).
\end{multline}
When taking the limit $K\rightarrow\infty$, as $\eta_{i}$ for $i=0,1,\dots,K$
under $\tilde{\mathbb{P}}$, is a.s. finite, we get 
\begin{multline}
    \lim_{K\rightarrow\infty}\tilde{\mathbb{P}}\left(\left|\tau_{i}-\tau_{j}\right|\geq\varepsilon\;\mathrm{for\:all}\:(i,j),\:i\neq j,\eta_{0}\geq\eta_{(K)}\right)\\
    \leq \lim_{K\rightarrow\infty}\tilde{\mathbb{P}}\left(\eta_{(1)}+\left(K-1\right)\varepsilon\leq\eta_{0}\right)=0.
\end{multline}
Now, on the event $\{\eta_{(k-1)}\leq\eta_{0}<\eta_{(K)}\}$, we have
that $\tau_{(K)}=\eta_{0}$ and $\tau_{(1)}=\eta_{(1)}$. By a similar
reasoning as above, 
\begin{multline}
\tilde{\mathbb{P}}\left(\left|\tau_{i}-\tau_{j}\right|\geq\varepsilon\;\mathrm{for\:all}\:(i,j),\:i\neq j,\eta_{(k-1)}\leq\eta_{0}<\eta_{(K)}\right)\\
\leq\tilde{\mathbb{P}}\left(\tau_{(K)}\geq\tau_{(1)}+\left(K-1\right)\varepsilon,\eta_{(k-1)}\leq\eta_{0}<\eta_{(K)}\right)\\
=\tilde{\mathbb{P}}\left(\eta_{0}\geq\eta_{(1)}+\left(K-1\right)\varepsilon,\eta_{(k-1)}\leq\eta_{0}<\eta_{(K)}\right)\\
\leq\tilde{\mathbb{P}}\left(\eta_{(1)}+\left(K-1\right)\varepsilon\leq\eta_{0}\right)\overset{K\rightarrow\infty}{\longrightarrow}0.
\end{multline}
The result follows after taking an expectation. The interchange in
the expectation and limit as $K\rightarrow\infty$ is justified as
$\tilde{\mathbb{P}}\left(\left|\tau_{i}-\tau_{j}\right|\geq\varepsilon\;\mathrm{for\:all}\:(i,j),\:i\neq j\right)$
is bounded by 1.
\end{proof}
Consequently, the number of G-SIBs in the economy needs to be restricted
by regulators to ensure that the probability of a market failure is
at an acceptable level. If the probability of a market failure as
implied by the number of existing G-SIBs is too high, then Theorem
\ref{Thm Market Failure Probability} enables the regulators to select
the number of G-SIBs such that the market failure probability is below
some given threshold. This implies, of course, that the excess G-SIBs
must be broken-up into smaller banks. 

Instead of breaking-up the G-SIBs, regulators can alternatively control
the probability of a market failure by requiring the existing G-SIBs
to change their asset/liability structures. This tool is discussed
in the next section.

\subsection{Changing the State of the Economy and Banks' Balance Sheets}

This section explores the impact of changing the state of economy
vector $X_{r}$ on the market failure probability. The idea, of course,
is that some of the inputs are under the control of the regulators,
e.g. required capital of a G-SIB. We investigate the impact of changes
in the initial conditions on the market failure probability. To facilitate
the exposition, let $X_{r}=(x_{1}(r),...,x_{d}(r))$, so that $\alpha_{i}(X_{r})=\alpha_{i}(x_{1}(r),...,x_{d}(r))$.

We redefine the default time for the $i^{th}$ G-SIB due to \emph{idiosyncratic
events} and the first time that a market-wide stress event occurs
in the following way: 
\begin{align*}
\eta_{i} & :=\inf\left\{ s:\alpha_{i}\left(X_{0}\right)+A_{i}(s)\geq Z_{i}\right\} \text{ for }i=1,\dots,K\\
\eta_{0} & :=\inf\left\{ s:\alpha_{0}\left(X_{0}\right)+A_{0}(s)\geq Z_{0}\right\} .
\end{align*}
Just as before, the default time of the $i^{th}$ G-SIB is 
\[
\tau_{i}=\min\left(\eta_{0},\eta_{i}\right).
\]
Then, it is easy to check that
\begin{align*}
\mathbb{P}\left(\eta_{i}>t|\left(X_{u}\right)_{0\leq u\leq t}\right) & =\exp\left(-\alpha_{i}\left(X_{0}\right)-A_{i}(t)\right)\\
\mathbb{P}\left(\tau_{i}>t|\left(X_{u}\right)_{0\leq u\leq t}\right) & =\exp\left(-\alpha_{i}\left(X_{0}\right)-\alpha_{0}\left(X_{0}\right)-A_{i}(t)-A_{0}(t)\right).
\end{align*}
To ensure that the probability distributions of $\eta_{i}$ and $\tau_{i}$
are correctly defined, we assign a positive probability to the event
$\left\{ \eta_{i}=0\right\} $ such that
\[
\mathbb{P}\left(\eta_{i}=0|X_{0}\right)=1-\exp\left(-\alpha_{i}\left(X_{0}\right)\right).
\]
This implies that 
\[
\mathbb{P}\left(\tau_{i}=0|X_{0}\right)=1-\exp\left(-\alpha_{i}\left(X_{0}\right)-\alpha_{0}\left(X_{0}\right)\right).
\]
The interpretation is that there is a positive probability of an ``instantaneous''
default at $t=0$. Under these modifications, we have the following
result.
\begin{thm}
(Comparative Statics)
\begin{multline}
    \frac{\partial}{\partial x_{\ell}(0)}\mathbb{P}\left(\left|\tau_{i}-\tau_{j}\right|<\varepsilon\text{ for some }(i,j)\in\left(1,...,K\right)\times\left(1,\dots,K\right),{i\neq j}\right)=\\
    \mathbb{E}\left[\left(\sum\limits _{i=0}^{K}\frac{\partial\alpha_{i}(X_{0})}{\partial x_{\ell}(0)}\right)\tilde{\mathbb{P}}\left(\left|\tau_{i}-\tau_{j}\right|\geq\varepsilon\text{ for all }(i,j),{i\neq j}\right)\right]=\\
    \mathbb{E}\Biggl[\left(\sum\limits _{i=0}^{K}\frac{\partial\alpha_{i}(X_{0})}{\partial x_{\ell}(0)}\exp\left(-\sum\limits _{i=0}^{K}\alpha_{i}(X_{0})\right)\right)\sum_{j\in P_K}\int\limits _{0}^{\infty}\hat{f}_{j_{1}}(x_{1})\int\limits _{x_{1}+\varepsilon}^{\infty}\hat{f}_{j_{2}}(x_{2})\\
    \int\limits _{x_{2}+\varepsilon}^{\infty}\hat{f}_{j_{3}}(x_{3})\dots \int\limits _{x_{K-2}+\varepsilon}^{\infty}\hat{f}_{j_{K-1}}(x_{K-1})\exp\left[-A_{j_{K}}(x_{K-1}+\varepsilon)-A_{0}(x_{K-1}+\varepsilon)\right]\\
    \,dx_{K-1} \dots \,dx_{3}\,dx_{2}\,dx_{1}\Biggr]
\end{multline}
where  $\hat{f}_{j_{k}}(x):={\alpha_{j_{k}}(X_{x})}\exp\left[-A_{j_{k}}(x)\right]$\footnote{It is worth noting that $\hat{f}_{j_k}$ is not the density of $\eta_{j_k}$. In this case, the density is $f_{j_k} (x)= \alpha_{j_{k}}(X_{x})\exp\left[-\alpha_{j_k}(X_0) - A_{j_{k}}(x)\right] $.
Hence, the change in notation.}.
The derivative is taken with respect to the $\ell$ component of the $\mathbb{R}^d$ vector $X$ at time $t=0$, i.e., 
we are analyzing the market failure probability when changing the initial state of the $\ell$ component of the
economy vector.
\end{thm}

\begin{proof}
First note that 
\begin{multline*}
    \frac{\partial}{\partial x_{\ell}(0)}\mathbb{P}\left(\left|\tau_{i}-\tau_{j}\right|<\varepsilon\text{ for some }(i,j), i \neq j \right)=\\
    \frac{\partial}{\partial x_{\ell}(0)}\mathbb{E}\left[\tilde{\mathbb{P}}\left(\left|\tau_{i}-\tau_{j}\right|<\varepsilon\text{ for some }(i,j), i \neq j\right)\right].
\end{multline*}
As $\tilde{\mathbb{P}}\left(\left|\tau_{i}-\tau_{j}\right|<\varepsilon\text{ for some }(i,j), i \neq j\right)\leq1$,
we can interchange expectation and derivative and so, it suffices
to find 
\[
\frac{\partial}{\partial x_{\ell}(0)}\tilde{\mathbb{P}}\left(\left|\tau_{i}-\tau_{j}\right|<\varepsilon\text{ for some }(i,j), i \neq j\right)
\]
and then take an expectation.

Now, with the change of definition of $\eta_{i}$ and $\eta_{0}$,
by a similar fashion as in Theorem \ref{Thm Market Failure Probability},
we get that 
\begin{multline*}
    \tilde{\mathbb{P}}\left(\left|\tau_{i}-\tau_{j}\right|<\varepsilon\text{ for some }(i,j), i \neq j\right)=\\
    1-\tilde{\mathbb{P}}\left(\left|\tau_{i}-\tau_{j}\right|\geq\varepsilon\text{ for all }(i,j),{i\neq j}\right)=\\
    1-\exp\left(-\sum\limits _{i=0}^{K}\alpha_{i}(X_{0})\right)\sum_{j\in P_K}\int\limits _{0}^{\infty}\hat{f}_{j_{1}}(x_{1})\int\limits _{x_{1}+\varepsilon}^{\infty}\hat{f}_{j_{2}}(x_{2})\int\limits _{x_{2}+\varepsilon}^{\infty}\hat{f}_{j_{3}}(x_{3})\dots\\
    \int\limits _{x_{K-2}+\varepsilon}^{\infty}\hat{f}_{j_{K-1}}(x_{K-1})\exp\left[-A_{j_{K}}(x_{K-1}+\varepsilon)-A_{0}(x_{K-1}+\varepsilon)\right]\,dx_{K-1}\dots\,dx_{1}.
\end{multline*}

Differentiating the previous equation with respect to $x_{\ell}(0)$,
we obtain 
\begin{multline*}
    \frac{\partial}{\partial x_{\ell}(0)}\tilde{\mathbb{P}}\left(\left|\tau_{i}-\tau_{j}\right|<\varepsilon\text{ for some }(i,j), i \neq j\right)\\
    =-\frac{\partial}{\partial x_{\ell}(0)}\tilde{\mathbb{P}}\left(\left|\tau_{i}-\tau_{j}\right|\geq\varepsilon\text{ for all }(i,j),{i\neq j}\right)\\
    =\left(\sum\limits _{i=0}^{K}\frac{\partial\alpha_{i}(X_{0})}{\partial x_{\ell}(0)}\right)\tilde{\mathbb{P}}\left(\left|\tau_{i}-\tau_{j}\right|\geq\varepsilon\text{ for all }(i,j),{i\neq j}\right)\\
    =\left(\sum\limits _{i=0}^{K}\frac{\partial\alpha_{i}(X_{0})}{\partial x_{\ell}(0)}\exp\left(-\sum\limits _{i=0}^{K}\alpha_{i}(X_{0})\right)\right)\sum_{j\in P_K}\int\limits _{0}^{\infty}\hat{f}_{j_{1}}(x_{1})\int\limits _{x_{1}+\varepsilon}^{\infty}\hat{f}_{j_{2}}(x_{2})\dots\\
    \int\limits _{x_{K-2}+\varepsilon}^{\infty}\hat{f}_{j_{K-1}}(x_{K-1})\exp\left[-A_{j_{K}}(x_{K-1}+\varepsilon)-A_{0}(x_{K-1}+\varepsilon)\right]\,dx_{K-1}\dots\,dx_{1}.
\end{multline*}
\end{proof}
Given estimates of the relevant intensities, these partial derivatives
are easily computed and they provide the information that regulators
can use to determine the impact of their regulatory restrictions on
the probability of a market failure.
\begin{rem}
(\emph{Linear Approximation})

For some simpler calculations, if $\alpha_{i}(X_{r})=\sum_{j=1}^{d}\beta_{ij}x_{j}(r)$
for $\beta_{ij}\in \mathbb{R}^+$ and $x_j(r)>0$ for $j=1,\dots, d$ and all $r>0$, then
\begin{multline}
\frac{\partial}{\partial x_{\ell}(0)}\mathbb{P}\left(\left|\tau_{i}-\tau_{j}\right|<\varepsilon\text{ for some }(i,j)\in\left(1,...,K\right)\times\left(1,\dots,K\right),{i\neq j}\right)=\\
\left(\sum\limits _{i=0}^{K}\beta_{i\ell}\right)\mathbb{P}\left(\left|\tau_{i}-\tau_{j}\right|\geq\varepsilon\text{ for all }(i,j),{i\neq j}\right).
\end{multline}
\end{rem}

\bibliographystyle{authordate1}
\bibliography{bibliography.bib}

\end{document}